\documentclass[11pt,a4paper]{article}
\usepackage{amsmath}
\usepackage[T1]{fontenc}
\usepackage[utf8]{inputenc}
\usepackage{amsfonts}
\usepackage{amsthm}
\usepackage{amstext}
\usepackage{amssymb}
\usepackage{color}
\usepackage{bbold}
\usepackage{hyperref}
\usepackage{url}

\newcommand{\cc}{\mathbb{C}}

\newcommand{\qq}{\mathbb{Q}}
\newcommand{\zz}{\mathbb{Z}}

\newcommand{\lh}{{\mathfrak h}}

\newcommand{\diag}{\operatorname{diag}}

\newcommand{\tr}{\operatorname{Tr}}

\newtheorem{theorem}{Theorem}

\newtheorem{lemma}[theorem]{Lemma}
\newtheorem*{lemma*}{Lemma}
\newtheorem{proposition}[theorem]{Proposition}

\newtheorem{remark}[theorem]{Remark}

\newtheorem*{open*}{Open~question}

\begin{document}

\title{Orbits of monomials and factorization into products of linear forms\thanks{The authors are supported by ANR project
CompA (code ANR--13--BS02--0001--01). Email: {\tt pascal.koiran@ens-lyon.fr}, {\tt ressayre@math.univ-lyon1.fr}}\\
}

\author{Pascal Koiran\\
Universit\'e de Lyon, Ecole Normale Sup\'erieure de Lyon, LIP\thanks{UMR 5668 ENS  Lyon, CNRS, UCBL.}\\
Nicolas Ressayre\\
Univ Lyon, Universit\'e Claude Bernard Lyon 1,\\
Institut Camille Jordan (CNRS UMR 5208),\\
43 blvd. du 11 novembre 1918, F-69622 Villeurbanne cedex, France}

\maketitle

\begin{abstract}
  This paper is devoted to the factorization of multivariate polynomials into
  products of linear forms, a problem which has applications to  differential
  algebra, to the resolution of systems of polynomial equations
  and to Waring decomposition
  (i.e., decomposition in sums of $d$-th powers
  of linear forms; 
    this problem is also known as
   {\em symmetric tensor decomposition}).
   We provide three black box algorithms for this problem.
  
  Our main contribution is an algorithm
    motivated by the application to Waring decomposition.
  This algorithm reduces the corresponding factorization problem to simultaenous
  matrix diagonalization, a standard task in linear algebra.
  The algorithm relies on ideas from invariant theory, and more specifically
  on Lie algebras.

{Our second algorithm reconstructs a factorization from several bivariate
  projections. Our  third algorithm reconstructs it from the determination
  of the zero set of the input polynomial, which is a union of hyperplanes.}
\end{abstract}

\section{Introduction}

The main contribution of this paper is a simple algorithm which determines 
whether an input polynomial $f(x_1,\ldots,x_n)$ has a factorization of the form
\begin{equation} \label{problem}
f(x)=l_1(x)^{\alpha_1} \cdots l_n(x)^{\alpha_n}
\end{equation}
where the linear forms $l_i$ are linearly independent.
The algorithm outputs such a factorization if there is one.
Our algorithm works in the black box model:
we assume that we have access to the input polynomial
$f$ only through a  ``black box'' 
which on input $(x_1,\ldots,x_n)$ outputs $f(x_1,\ldots,x_n)$.

We therefore deal with a
very special case of the polynomial factorization problem.
As explained in Section~\ref{waring} below, this special case already has an
  interesting application to Waring decomposition.
The algorithm is based on (elementary) ideas of invariant theory,
but is nonetheless quite simple: it essentially boils down to the
simultaneous diagonalization of commuting matrices, a standard task
in linear algebra.
For the general problem of factorization in the black box model there
is a rather involved algorithm by Kaltofen and Trager~\cite{KalTra90},
see Section~\ref{comparison} for more details.
Our factorization  algorithm  seems to be the first to rely on ideas from invariant theory, and to reduce a multivariate polynomial factorization problem
to matrix
diagonalization.
Let us now explain why it is natural to use invariant theory in this context.

\subsection{Connection with invariant theory} \label{gct}

Consider a field $K$ of characteristic 0
and a polynomial $f \in K[x_1,\ldots,x_n]$.
By definition, the orbit $\mathrm{Orb}(f)$
of $f$ under the action of the general linear group is
the set of polynomials of the form $f(A.x)$ where $A \in GL_n(K)$
is an arbitrary invertible matrix.
In their Geometric Complexity Theory program~\cite{mulmuley01,mulmuley08},
Mulmuley and Sohoni
have proposed the following approach to lower bounds in algebraic complexity:
in order to prove a lower bound for a polynomial~$g$, show that it
does not belong to a suitable {\em orbit closure}~$\overline{\mathrm{Orb}(f)}$.
The case where $f$ is the determinant polynomial is of particular interest
as it allows to address the infamous ``permanent versus determinant''
problem. Mulmuley and Sohoni have also proposed a specific
representation-theoretic approach to deal with this orbit closure problem.
As it turns out, the representation-theoretic approach provably
does not work~\cite{burgisser16}.
The general approach based on orbit closure remains plausible, but has so far
not produced any major lower bound result because the orbit closure of
the determinant is difficult to describe.
By contrast, the renewed interest in invariant theory has led
to new positive results, i.e., to new polynomial time algorithms:
see for instance~\cite{burgisser13,BGOWW17,garg16,mulmuley12} and especially~\cite{kayal2012affine}, which
is a main inspiration for this paper.

We deal here with the simplest of all orbits, namely, the orbit
of a single monomial $x_1^{\alpha_1} \ldots x_n^{\alpha_n}$, and we derive
a new factorization algorithm.
It is immediate from the definition that this  orbit is the set of polynomials
that can be factorized as in~(\ref{problem}) with linearly independent
forms.
Note that the orbit closure of the monomial $x_1x_2\ldots x_n$
is the set of polynomials that can be written as products of $n$ linear forms
(without any assumption of linear independence).
{%
  This is well known in algebraic geometry, see example (5) in Section~3.1.2 of~\cite{landsbergGCT} and exercise 3.1.4.2 in the same book.}
Moreover, equations for this orbit closure are known,
see chapter 9 of~\cite{landsbergGCT} for a derivation of the equations
and the history of this subject.
However, no factorization algorithm relying on ideas from invariant theory
is currently known for arbitrary products of linear forms.
We suggest this problem as a natural step before considering more complicated orbit closure problems.

\subsection{Application to Waring decomposition} \label{waring}

The factorization problem studied here is motivated mainly by an algorithm
due to Neeraj Kayal (see Section~5 of~\cite{kayal11}).
Factorization in products of linear forms is also useful
for algorithmic differential algebra~\cite{SingerUlmer97,vanHoeij99} and for the resolution of systems of algebraic equations 
by factorization of the $U$-resultant~\cite{cox06,Kobayashi88}.

Kayal's algorithm determines whether a homogeneous polynomial $f$
of degree $d$ in $n$ variables
can be written as the sum of $n$ $d$-th powers of linearly
independent forms.
This algorithm is based on the fact that such a polynomial has a Hessian
determinant which factors as a product of $(d-2)$-th powers
of linearly independent forms. In~\cite{kayal11} the Hessian is factorized with
Kaltofen's algorithm
{  for the factorization of polynomials given by straight-line programs}~\cite{kaltofen89}. The decomposition
of $f$ as a sum of $d$-th powers can be recovered from this information.
The algorithm presented in this note can therefore be used instead of
Kaltofen's algorithm to solve the same  decomposition problem.

Building on these ideas from~\cite{kayal11}, it was recently shown
in~\cite{GKP18} how to recover up to $O(n^2)$ terms in a Waring decomposition\footnote{This algorithm works when the linear forms to be recovered are sufficiently generic; efficient reconstruction in the worst case is still open.}
(and more generally in a sum of powers of affine forms with possibly
different exponents in each power).
The algorithm works for polynomials of degree $d \geq 5$ and is based
on the factorization of a ``generalized Hessian'' into products of linear forms.
There are now up to order $n^2$ distinct linear forms in the factorization,
and that many linear forms must of course  be linearly dependent.
This provides further motivation for the problem suggested at the end of
Section~\ref{gct} (namely, the extension of our algorithm
to the case of linearly dependent forms).
Factorization in products of dependent forms is discussed
at the end of Section~\ref{comparison}.

\subsection{Comparison with previous factorization algorithms}
\label{comparison}

As mentioned above, the algorithm for Waring decomposition in~\cite{kayal11} relies on Kaltofen's factorization algorithm~\cite{kaltofen89} which works
in the arithmetic circuit (or ``straight-line program'') model:
the input polynomial $f$ is described by an arithmetic circuit, and the output is a list of arithmetic circuits for 
the irreducible factors of $f$ together with their multiplicities.

One could instead appeal to the black-box factorization algorithm by
Kaltofen and Trager~\cite{KalTra90}. 
In this case, instead of factorizing a circuit for the determinant
of a Hessian matrix one would use a black box for 
the determinant of this matrix.
The algorithm from~\cite{KalTra90} produces a black box for the irreducible factors of $f$ given
a black-box for evaluating $f$.

Compared to~\cite{KalTra90,kaltofen89} our algorithm works in a hybrid model: we use the most general of the two 
for the input polynomial (black box representation) but we explicitly determine the linear forms $l_i$ in~(\ref{problem})
when they exist.\footnote{It would anyway be easy to explicitly determine the $l_i$ by interpolation from a black box
for these linear forms.} For the general polynomial factorization problem, it is apparently not known how to efficiently produce "small" arithmetic circuits for the irreducible factors of a polynomial $f$ given a black-box for $f$. 
Due to the black box algorithm of~\cite{KalTra90}, this would be equivalent to producing a small arithmetic circuit for a polynomial
given a black box for this polynomial.

The algorithms from~\cite{KalTra90,kaltofen89}
project the original  $n$-variate factorization problem to a bivariate
factorization problem, solve the bivariate problem using a factorization
algorithm for polynomials in dense representation, and then lift the result
to a factorization of the $n$-variate input polynomial.
It will be clear that our algorithm is based on a very different principle:
instead of projecting we do linear algebra computations
directly in $n$-dimensional space.

There is an intringuing connection between our 
algorithm and Gao's algorithm
for the absolute factorization of bivariate polynomials~\cite{gao03}: they are both
based on the study of certains partial differential equations.
For the connection of our approach to PDEs see Lemma~\ref{lieP}
in Section~\ref{invback}.

As explained in Section~\ref{waring}, for the application
to Waring decomposition following~\cite{kayal11}
we can assume that the linear forms $l_i$ are independent.
This assumption does not seem so natural in other applications
such as differential algebra~\cite{SingerUlmer97,vanHoeij99}
or the resolution of systems of polynomial equations~\cite{cox06,Kobayashi88}.
For this reason,  we present in Section~\ref{bivariate}
another algorithm for factorization into
products of linear forms based like~\cite{KalTra90,kaltofen89}
on bivariate projections. Our goal in that section
is to give a simpler algorithm
which takes advantage of the fact that we are considering only a special
case of the polynomial factorization problem.
We present another simple algorithm in Section~\ref{hyperplane}.
{%
  This  algorithm requires a univariate factorization algorithm, and the projection-based algorithm requires a bivariate factorization algorithm (see Sections~\ref{bivariate} and~\ref{hyperplane} for more details).}

For these last two algorithms, no assumption of linear independence is needed.
This is also the case for the algorithms in~\cite{Kobayashi88,vanHoeij99}.
In these two papers no complexity analysis is provided, and it is assumed
in the second one that the polynomial to be factorized is squarefree.
We note that the algorithm from~\cite{Kobayashi88} bears some
similarity to the algorithm that we present in Section~\ref{hyperplane}:
{  both are based on the determination of the zero set of
  the input polynomial, which is a union of hyperplanes.}

\subsection{On the choice of fields}

Polynomial factorization problems come with many variations.
In particular, the following choices need to be made:
\begin{itemize}
\item[(i)] The input is a polynomial $f \in K[x_1,\ldots,x_n]$.
  What field $K$ do we choose as field of coefficients for $f$?

\item[(ii)] What field $\mathbb{K}$ do we choose as field of coefficients
  for the output?
  More precisely, the output is a factorization $f=g_1 \ldots g_k$ where
  the polynomials $g_i$ belong to $\mathbb{K}[x_1,\ldots,x_n]$ for
  some field extension $\mathbb{K}$ of $K$,
  and are irreducible over $\mathbb{K}$.
  In the literature it is often
  (but not always) assumed that $K=\mathbb{K}$.

\item[(iii)] How do we represent the field elements? Assume for instance
  that $K=\qq$ and that we are interested in absolute factorization,
  i.e., factorization over $\mathbb{K} = \overline{\qq}$ (the algebraic closure
  of $\qq$). Do we insist on a symbolic representation for the coefficients
  of the $g_i$'s (in this case, the coefficients would be represented
  as elements of an extension of $\qq$ of finite degree) or, using an embedding
  $\overline{\qq} \subseteq \cc$, are we happy to compute only numerical
  approximations of these coefficients?
\end{itemize}
Absolute factorization
seems to be the most natural choice for this paper because of the application
to  Waring decomposition (this problem has been studied mostly in algebraically
closed fields\footnote{Some results are also known
  for the field of real numbers~\cite{carlini12,comon12}}).
Moreover, for any field $\mathbb{K}$
if a decomposition of $f$ of the form~(\ref{problem}) with the $l_i$
in $\mathbb{K}[x_1,\ldots,x_n]$ is possible then this decomposition
clearly is an absolute factorization of $f$.

  Nevertheless, we do not commit to any specific choice
for (i), (ii) and (iii) except that $K$ must be of characteristic zero.
This is possible because our main algorithm is a {\em reduction} (to matrix
diagonalization). Any (efficient) algorithm for this standard linear algebra
task for a specific choice of (i), (ii) and (iii)
will therefore yield an (efficient) factorization algorithm.
We elaborate on the complexity of our reduction in Section~\ref{complexity}.

\subsection{Complexity of our { invariant-theoretic} algorithm} \label{complexity}

The black box algorithm in~\cite{KalTra90} applies to polynomials with coefficients in a field $K$ of characteristic 0.
The only assumption on $K$ if that a factorization algorithm for univariate polynomials in $K[x]$ is available.
This black box algorithm can therefore be thought of as a reduction from multivariate to univariate polynomial factorization.
In order to evaluate precisely the complexity of this algorithm for a specific field $K$, one must of course take into account the complexity of the univariate factorization problem for this particular field.

Likewise, our main algorithm can be thought of as a reduction to (simultaneous) matrix diagonalization.\footnote{Note that diagonalizing a matrix is clearly related to the factorization of its characteristic polynomial.}
When we write that the algorithms of Section~\ref{factorization} run in polynomial time, we mean polynomial in $n$ (the number of variables of the input polynomial) and $d$ (its degree). In particular, the algorithm makes  $\mathrm{poly}(n,d)$ 
calls to the black box for $f$. It also performs simultaneous diagonalization on $n$ (commuting) matrices, and makes a few other 
auxiliary computations. The main one is the determination of the Lie algebra of $f$,  which as explained in Section~\ref{invback} is a linear algebra problem; a polynomial black box algorithm for it can be found in~\cite{kayal2012affine}.
{%
  A more precise analysis of our algorithm can be found in
  the appendix.
  It suggests that the computation of the Lie algebra of $f$ is
  a particularly expensive step.
  Improving the algorithm from~\cite{kayal2012affine} (or its analysis
  in the appendix) seems to be an interesting open problem.}

{%
  If we just want to {\em decide} the existence of a suitable
  factorization (rather than compute it) our algorithm becomes purely algebraic,
  i.e., it just performs arithmetic operations (additions, multiplications and tests to zero) on the function values given by the black box for $f$.
In particular, we do not need to factor univariate polynomials or diagonalize matrices.}
  
Like in~\cite{kaltofen89,KalTra90} our algorithm is randomized and can return a wrong answer with a small probability $\epsilon$. This is unavoidable 
because  homogeneous polynomials of degree $d$ in $n$ variables have $\binom{n+d-1}{d}$ coefficients and this  is bigger than any fixed polynomial in $n$ and $d$ if these two parameters are  nonconstant. As a result, for a polynomial $f$ of form~(\ref{problem}) there will always be another polynomial $g$ which agrees with $f$ on all points queried on input $f$.
The algorithm will therefore erroneously\footnote{Indeed, the algorithm should report failure if $g$ is not of form~(\ref{problem}), or if it is should return a different factorization than for $f$.} output the same answer on these two inputs. The probability of error $\epsilon$  can be thought of as a small fixed constant, and as usual it can be made as small as desired by repeating the algorithm (or by changing the parameters in the algorithm from~\cite{kayal2012affine} for the 
computation of the Lie algebra; this is the main source of randomness in our algorithm\footnote{%
   If $f$ is given explicitly as a sum of monomials,
  the Lie algebra can be computed deterministically in polynomial time; this is
  clear from the characterization of the Lie algebra in Lemma~\ref{lieP}.}).

\subsection{Organization of the paper}

In Section~\ref{background} we recall some background
on matrix diagonalization, simultaenous diagonalization
and invariant theory.
In Section~\ref{orbitsection} we give a characterization of the polynomials
in the orbit of a monomial.
We use this characterization in  Section~\ref{factorization} to derive
our main 
{%
  algorithm for factorization into products of (independent) linear forms.}
An algorithm based on the older idea of bivariate projections is presented
in Section~\ref{bivariate}.
{  In contrast to~\cite{kaltofen89,KalTra90} this
algorithm recovers a factorization of the input polynomial from {\em several}
bivariate projections.}
  Another simple algorithm is presented in Section~\ref{hyperplane}.
  {  As mentioned earlier, this algorithm relies on the determination
    of the zero set of $f$.}
{%
 Our last two algorithms do not rely on any invariant theory and do not require any independence property for the linear forms. As pointed out at the end of Section~1.1, for factorization into products of arbitrary linear forms no algorithm
that would rely on ideas from invariant theory is known at this time.}

{
  The paper ends with two appendices where we analyze the complexity of our
  three algorithms in more detail than in the main body. 
    In particular, we point out in Appendix~\ref{bb} an  optimization
  of our invariant-theoretic algorithm for the ``white box'' model,
  in which the black box for $f$ is implemented by an arithmetic circuit.}

\section{Background} \label{background}

We first recall the Schwarz-Zippel lemma~\cite{Schw,zippel},
a ubiquitous tool in the analysis
of randomized algorithms.
\begin{lemma}
  Let $f \in K[x_1,\ldots,x_n]$ be a nonzero polynomial.
  If $a_1,\ldots,a_n$ are drawn independently and uniformly at random
  from a finite set $S \subseteq K$ then
  $$\Pr[f(a_1,\ldots,a_n)=0] \leq \deg(f)/|S|.$$
\end{lemma}
A similar result with a slightly worse bound was obtained a little earlier
by DeMillo and Lipton~\cite{demillo}.
In the remainder of this section we recall some background on
matrix diagonalization, on simultaenous diagonalization, on invariant theory
and Lie algebras.

\subsection{Background on matrix diagonalization} \label{diagback}

Since our main algorithm is a reduction to matrix diagonalization,
it is appropriate to provide some brief background on the algorithmic
solutions to this classical problem. After a first course on linear algebra, this might look like a simple task:
to diagonalize a matrix $M$, first compute its eigenvalues. Then, for each eigenvalue $\lambda$ compute a basis
of $\mathrm{Ker}(M-\lambda.I)$. But this problem is more subtle than it
{  seems at first sight}.

Let us begin with numerical algorithms. There is a
vast literature on numerical methods for eigenvalue problems
(see for instance~\cite{bjorck16} and the references there).
Naively, one might want to compute
the eigenvalues of  $M$ by computing the roots of its characteristic
polynomial $\chi_M(\lambda)=\det(M-\lambda I)$. This approach is hardly ever
used in practice for large matrices because the roots of a polynomial
can be very sensitive to perturbations of its coefficients~\cite{wilkinson84}.
A theoretical analysis explaining why such a bad behaviour is rather prevalent
can be found in~\cite{burgisser17}.
The QR algorithm is now considered to be the standard algorithm for computing
all eigenvalues and eigenvectors of a dense matrix~\cite{bjorck16}.
It works well in practice, but a thorough understanding of this algorithm
(or of {\em any} efficient and stable
numerical algorithm for the computation of eigenvalue -- eigenvector pairs)
is still lacking, see Open Problem~2 in~\cite{burgisser2013}.

Let us now turn to symbolic methods.
In the absence of roundoff errors, an approach based on the computation
of the characteristic polynomial becomes feasible
(see~\cite{pernet07} for
the state of the art on the computation of
this polynomial). From the knowledge of $\chi_M$ we can decide whether
$M$ is diagonalizable using the following classical result from linear algebra.
\begin{proposition} \label{diagmat}
  Let $K$ be a field of characteristic 0 and let 
  $\chi_M$ be the characteristic polynomial of a matrix $M \in M_n(K)$.
  Let $P_M = \chi_M / \mathrm{gcd}(\chi_M,\chi_M')$ be the squarefree part of
  $\chi_M$. The matrix $M$ is diagonalizable over $\overline{K}$ iff
  $P_M(M)=0$.\footnote{%
    An equivalent characterization is
    that the minimal polynomial of $M$ has only simple roots.}
    Moreover, in this case $M$ is diagonalizable over $K$ iff all the
  roots of $P_M$ lie in $K$.
\end{proposition}
Once we know that $M$ is diagonalizable, computing the diagonal form of
$M$ symbolically requires the factorization of $P_M$.
We note that for $K=\qq$, finding the roots of $P_M$ in $\qq$
is cheaper than the general problem of factorization in irreducible
factors over $\qq[X]$~(\cite{aecf}, Proposition~21.22).
This faster algorithm should therefore be used to diagonalize over $\qq$.
For the purpose of this paper, this is relevant for  factorisation into  a product of linear forms with rational coefficients.

Once we know the eigenvalues of $M$ and their multiplicities, the last step
is the computation of a transition matrix $T$ such that $T^{-1}MT$ is diagonal.
For this step we refer to~\cite{giesbrecht94,giesbrecht95,villard97}.
These papers consider the more general problem
of computing symbolic representations of the Jordan normal form.

The knowledge of $T$ is particularly important for the application
to factorization into product of linear forms because (as shown in
Section~\ref{factorization}) these forms
can be read off directly from the transition matrix.
If we just want to know whether such a factorization is possible over $K$
or $\overline{K}$, Proposition~\ref{diagmat} is sufficient.

\subsection{Simultaneous diagonalization} \label{simulback}

It is a well known fact of linear algebra that a family of diagonalizable matrices
is simultaneously diagonalizable if and only if they pairwise commute. 
We will use this criterion to test whether a family of matrices $A_1,\ldots,A_k$ is simultaneously diagonalizable.
If the test succeeds, we will then need to diagonalize them. Note that a transition matrix which diagonalizes
$A_1$ may not necessarily diagonalize the other matrices (this may happen if $A_1$ has an eigenvalue of multiplicity 
larger than 1). We can nonetheless perform a simultaneous diagonalization by diagonalizing a single matrix. Indeed, as suggested in
Section~6.1.1 of~\cite{kayal2012affine} we can diagonalize a random linear combination of the $A_i$'s. We sketch a proof of this simple fact below.
For notational simplicity we consider only the case of two matrices.
The general case can be treated in a similar way.
\begin{lemma} \label{simulemma}
Assume that $M,N \in M_n(k)$ are two simultaneously diagonalizable matrices. There is a set $B \subseteq K$  of size
at most $n(n-1)/2$ such that for any $t \in K \setminus B$ any eigenvector of $M+tN$ is also an eigenvector of $M$ and~$N$.
\end{lemma}
\begin{proof}
Since $M$ and $N$ are simultaneously diagonalizable we may as well work in a basis where these matrices become diagonal.
We therefore assume without loss of generality that $M=\diag(\lambda_1,\ldots,\lambda_n)$ and 
$N=\diag(\mu_1,\ldots,\mu_n)$. We then have $M+tN=\diag(\lambda_1+t\mu_1,\ldots,\lambda_n+t\mu_n)$ for any $t \in K$. We may take for  $B$  the set of $t$'s such that $\lambda_i+t\mu_i = \lambda_j + t \mu_j$ for some pair $\{i,j\}$
such that $(\lambda_i,\mu_i) \neq (\lambda_j,\mu_j)$. This is indeed a set of size at most $n(n-1)/2$,
and for $t {\not \in} B$ the eigenspace of $M+tN$ associated to the eigenvalue $\lambda_i+t\mu_i$ is the intersection of the eigenspace of $M$ associated to $\lambda_i$ and of the eigenspace of $N$ associated to $\mu_i$.
In particular, any  eigenvector of $M+tN$ is also an eigenvector of $M$ and $N$.
\end{proof}

\begin{proposition}
  Assume that $M,N \in M_n(k)$ are two simultaneously diagonalizable matrices
  and that $t$ is drawn from the uniform distribution on a finite
  set $S \subset K$. With probability at least $1-\frac{n(n-1)}{2|S|}$,
  all the transition matrices which diagonalize $M+tN$ also diagonalize $M$
  and $N$.
\end{proposition}
\begin{proof}
  We show that the required property holds true for any $t$
  that does not belong to the ``bad set'' of Lemma~\ref{simulemma}.
  
  For an invertible matrix $T$, $T^{-1}(M+tN)T$ is diagonal iff
  all the column vectors of $T$ are eigenvectors of $M+tN$.
  But for $t {\not \in B}$, any eigenvector of $M+tN$ is also  an eigenvector
  of $M$ and $N$. As a result, if $T^{-1}(M+tN)T$ is diagonal then
  $T^{-1}MT$ and $T^{-1}NT$ are diagonal as well.
  \end{proof}

  \subsection{Background on invariants and Lie algebras} \label{invback}

  In this section and in the remainder of the paper, $K$ denotes a field
  of characteristic 0.
The general linear group $GL_n$ acts on the polynomial ring $K[x_1,\ldots,x_n]$
by linear change of variables: an invertible matrix $A \in GL_n$ sends
a polynomial $P(x) \in K[x_1,\ldots,x_n]$ to $P(A.x)$.
The group of invariant of $P$ is the group of matrices $A$ such that $P(A.x) = P(x)$.
We recall that this is a Lie group. Its Lie algebra $\mathfrak{g}$ is a linear subspace of
$M_n(K)$ defined as the tangent space of $G$ at identity.
More precisely, $\mathfrak{g}$ is the ``linear part'' of the tangent space;
the (affine) tangent space is $I+\mathfrak{g}$.

The Lie algebra associated to the group of invariants of $P$ will be
called simply ``Lie algebra of $P$'', and we will denote it by $\mathfrak{g}_P$.
It can be explicitly computed as follows.
\begin{lemma}[Claim 59 in~\cite{kayal2012affine}] \label{lieP}
  A matrix $A=(a_{ij}) \in M_n(K)$ belongs to the Lie algebra of $P$
  if and only if
  \begin{equation} \label{liePeq}
    \sum_{i,j \in [n]} a_{ij} x_j \frac{\partial P}{\partial x_i}=0
    \end{equation}
\end{lemma}
The elements of the Lie algebra therefore correspond to linear dependence
relations between the polynomials $x_j \frac{\partial P}{\partial x_i}$.

As an example we determine the group of invariants of monomials.
\begin{lemma} \label{moninvar}
  The group of invariants of a monomial $m=x_1^{\alpha_1}....x_n^{\alpha_n}$
  with $\alpha_i \geq 1$ for all $i$
  is generated by:
  \begin{itemize}
  \item[(i)] The diagonal matrices $\diag(\lambda_1,\ldots,\lambda_n)$
    with $\prod_{i=1}^n \lambda_i^{\alpha_i} =1$.
We denote this subgroup of $GL_n$ by $T_{\alpha}$, where $\alpha$  is the tuple $(\alpha_1,\ldots,\alpha_n)$.
    \item[(ii)] The permutation matrices which map any variable $x_i$ to a variable $x_{\pi(i)}$ with same exponent in $m$ (i.e., with $\alpha_i=\alpha_{\pi(i)}$).
    \end{itemize}
\end{lemma}
\begin{proof}
  The monomial is obviously invariant under the actions of matrices from
  (i) and (ii). Conversely, assume that $m$ is invariant under the action
  of an invertible matrix $A$. By uniqueness of factorization, $A$ must send every variable $x_i$ to the multiple of another variable, i.e., to
  $\lambda_ix_{\pi(i)}$. Moreover we must have $\alpha_i=\alpha_{\pi(i)}$ and
  $\prod_{i=1}^n \lambda_i^{\alpha_i} =1$, so $A$ is in the group generated by (i) and
  (ii).
\end{proof}
The Lie algebras of monomials is determined in Proposition~\ref{liemon}.
In this paper we will follow the Lie-algebraic approach
from~\cite{kayal2012affine}.
As a result we will not work directly with groups of invariants.

If two polynomials are equivalent under the action of $GL_n$, their Lie algebras are conjugate. More precisely:
\begin{proposition}[Proposition 58 in~\cite{kayal2012affine}] \label{lieconj}
  If $P(x)=Q(A.x)$ then $$\mathfrak{g}_P = A^{-1}.\mathfrak{g}_Q.A$$
\end{proposition}

\section{The orbit of a monomial} \label{orbitsection}

Throughout the paper, $m$ denotes a monomial
$x_1^{\alpha_1}....x_n^{\alpha_n}$  with all exponents $\alpha_i \geq 1$.
\begin{proposition} \label{liemon}
  The Lie algebra $\mathfrak{g}_m$ of a monomial $m=x_1^{\alpha_1} \cdots x_n^{\alpha_n}$ with all exponents $\alpha_i \geq 1$ is the space of diagonal matrices $\diag(\lambda_1,\ldots,\lambda_n)$
  such that $\sum_{i=1}^n \alpha_i \lambda_i =0$.
\end{proposition}
\begin{proof}
  By Lemma~\ref{lieP},  all these matrices are in $\mathfrak{g}_m$ since
  $m$ satisfies the equation $x_i \frac{\partial m}{\partial x_i} = \alpha_i m$.
  Conversely, if $A \in \mathfrak{g}$ all off-diagonal entries $a_{ij}$
  must vanish since the monomial $x_j \frac{\partial m}{\partial x_i}$
  could not cancel with any other monomial in~(\ref{liePeq}).
\end{proof}
\begin{remark} \label{liemongen}
  The above characterization of  $\mathfrak{g}_m$ is no longer true  if some exponents $\alpha_i$ may vanish. Indeed, in this case there is no constraint on the entries in row $i$ of a matrix in $\mathfrak{g}_m$.
  However, we note for later use that in all cases,
  the space of diagonal matrices $\diag(\lambda_1,\ldots,\lambda_n)$
  which lie in  $\mathfrak{g}_m$
is defined by  $\sum_{i=1}^n \alpha_i \lambda_i =0$.
\end{remark}
It is easy to check by a direct computation that the Lie algebra determined
in Proposition~\ref{liemon} is (as expected) equal to the tangent space at identity
of the group $T_{\alpha}$ from Lemma~\ref{moninvar}.
The next result turns Proposition~\ref{liemon} into an equivalence.
\begin{proposition} \label{moncharlie}
  Let $f \in K[x_1,\ldots,x_n]$ be a homogeneous polynomial of degree $d$.
  The two following properties are equivalent:
  \begin{itemize}
\item[(i)]  $f$ is a monomial which depends on all of its $n$ variables.
\item[(ii)] The Lie algebra of $f$ is an $(n-1)$-dimensional subspace of the space of diagonal matrices.
\end{itemize}
\end{proposition}
\begin{proof}
  We have seen in Proposition~\ref{liemon} that (i) implies (ii).
Conversely, for any polynomial $P$ let us denote by $\mathfrak{d}_P$ the subspace of its Lie algebra made of diagonal matrices.
By Lemma~\ref{lieP}, $\mathfrak{d}_f$ is the space of  of matrices $\diag(\lambda_1,\ldots,\lambda_n)$ 
such that 
\begin{equation} \label{diagf}
\sum_{i=1}^n \lambda_i x_i \frac{\partial f}{\partial x_i}=0
\end{equation}
 For any monomial $m$, $x_i \frac{\partial m}{\partial x_i}$ is proportional to $m$. This implies that $\mathfrak{d}_f$ is the intersection of the  $\mathfrak{d}_m$'s for the various
monomials $m$ appearing in $f$ since the contributions to~(\ref{diagf}) coming from different monomials cannot cancel.
By Remark~\ref{liemongen}, for two distinct monomials $m_1$ and $m_2$ appearing in $f$ the subspaces  $\mathfrak{d}_{m_1}$ and   $\mathfrak{d}_{m_2}$
are distinct since they are defined by linear forms that are not proportional
(here we use the homogeneity of $f$). It follows that their intersection is of dimension $n-2$ in contradiction with (ii). 
Therefore, only one monomial can appear in $f$.
Finally, by Remark~\ref{liemongen} all of the $n$ variables must appear
in this monomial; otherwise, $\mathfrak{g}_f$
would contain some nondiagonal matrices.
\end{proof}
We can now 
characterize the Lie algebras of polynomials in the orbit of a monomial.

\begin{theorem} \label{orbit}
  Consider a monomial $m=x_1^{\alpha_1} \cdots x_n^{\alpha_n}$
  with $\alpha_i \geq 1$ for all $i$, a homogeneous
  polynomial $f \in K[x_1,\ldots,x_n]$ of degree $d=\alpha_1+\cdots+\alpha_n$ and an invertible matrix~$A$.
  The two following properties are equivalent.
  \begin{itemize}
  \item[(i)] The action of $A$ sends $m$ to a multiple of $f$, i.e.,   $m(A.x)=c.f(x)$ for some constant $c$.

  \item[(ii)] The Lie algebras of $f$ and $m$ are conjugate by $A$, i.e.,
    $\mathfrak{g}_f = A^{-1}.\mathfrak{g}_m.A$.
   \end{itemize}
\end{theorem}
\begin{proof}
  Proposition~\ref{lieconj} shows that (i) implies (ii).
  For the converse, assume that $\mathfrak{g}_f = A^{-1}.\mathfrak{g}_m.A$
  and define $g(x)=f(A^{-1}.x)$. By Proposition~\ref{lieconj} we have
  $\mathfrak{g}_g = \mathfrak{g}_m$.
It follows from Propositions~\ref{liemon} and~\ref{moncharlie}  
 that $g=\lambda.m$ for some nonzero constant $\lambda$. 
We therefore have $m(Ax)=f(x)/\lambda$.
 \end{proof}
This characterization takes a particularly simple form in the case
of equal exponents.
\begin{theorem} \label{equalexp}
  Consider a monomial $m=(x_1 \cdots x_n)^{\alpha}$ and a homogeneous
  polynomial $f \in K[x_1,\ldots,x_n]$ of degree $d=n\alpha$.
  The two following properties are equivalent.
  \begin{itemize}
  \item[(i)] Some multiple of $f$ belongs to the orbit of $m$, i.e., $m(A.x)=c.f(x)$ for some invertible matrix $A$ and some constant $c$.

  \item[(ii)] The Lie algebra of $f$ has a basis made of $n-1$
    diagonalizable matrices of trace zero which pairwise commute.
  \end{itemize}
  Moreover, $f$ is a constant multiple of $m$ if and only its Lie algebra
  is the space of diagonal matrices of trace zero.
\end{theorem}
\begin{proof}
Let $f$ be in the orbit of $m$. By Proposition~\ref{lieconj}, in order to establish~(ii) for $f$ we just need to check that this property is true for $m$. This is the case since (by Proposition~\ref{liemon})  the Lie algebra of $m$ is the space of diagonal matrices of trace 0.

Conversely, assume that (ii) holds for $f$. It is a well known fact of linear algebra that a family of diagonalizable matrices
is simultaneously diagonalizable if and only if they pairwise commute. By simultaneously diagonalizing the $n-1$ matrices in the basis of $\mathfrak{g}_{f}$ we find that this Lie algebra is conjugate to $\mathfrak{g}_m$ (which as we just saw  is the space of diagonal matrices of trace 0).  Hence some constant multiple of $f$ is in the orbit of $m$ by Theorem~\ref{orbit}.

As to the second part of the theorem, we have already seen that 
$\mathfrak{g}_m$ is the space of diagonal matrices of trace zero.
Conversely, if $\mathfrak{g}_f = \mathfrak{g}_m$ we can apply
Theorem~\ref{orbit} with $A = \mathrm{Id}$ and it follows that $f$ is
a constant multiple of~$m$.
\end{proof}
Note that if (ii) holds for some basis of  $\mathfrak{g}_{f}$ this property holds for all bases.
Also, if $K$ is algebraically closed we can always take $c=1$ in Theorems~\ref{orbit} and~\ref{equalexp}.

\section{Factorization into products of independent forms} \label{factorization}

By definition, the orbit of a monomial $m=x_1^{\alpha_1} \cdots x_n^{\alpha_n}$ contains the polynomial $f$ if an only if $f$ can be written as $f(x)=l_1(x)^{\alpha_1} \cdots l_n(x)^{\alpha_n}$ where the linear forms $l_i$ are linearly independent.
We will exploit the characterization of orbits obtained in Section~\ref{orbitsection} to factor such polynomials.
We assume that we have access to a black-box for $f$.
We begin with the simpler case of equal exponents.
Note that this is exactly what is needed in Section~5 of~\cite{kayal11}.

\subsection{Equal exponents} \label{equal}

In this section we describe an algorithm which takes as input a homogeneous polynomial $f \in K[x_1,\ldots,x_n]$ of degree $d=n\alpha$, determines if it can be expressed as  $f = (l_1 \cdots  l_n)^{\alpha}$ where the $l_i$'s are linearly independent forms
and finds such a factorization if it exists.
In the first three steps of the following algorithm
we decide whether such a factorization
exists over $\overline{K}$, 
and in the last two we actually compute the
factorization.
\begin{enumerate}
\item Compute a basis $B_1,\ldots,B_k$ of the  Lie algebra of $f$.
  
\item Reject if $k \neq n-1$, i.e.,
  if the Lie algebra is not of dimension $n-1$.

\item Check that the matrices $B_1,\ldots,B_{n-1}$ commute,
  are all diagonalizable over $\overline{K}$ 
and of trace zero.
  If this is the case, declare that $f$ can be factored as
  $f = (l_1 \cdots  l_n)^{\alpha}$ 
 where the $l_i$'s are
  linearly independent forms.
  Otherwise, reject.

\item Perform a simultaneous diagonalization of the $B_i$'s, i.e.,
  find an invertible matrix $A$ such that the $n-1$ matrices
  $AB_iA^{-1}$ are diagonal.

\item At the previous step we have found a matrix $A$ such that
  $f(A^{-1}x) = \lambda.m(x)$ where $m$ is the monomial  $(x_1 \cdots  x_n)^{\alpha}$. We therefore have $f(x)=\lambda.m(Ax)$ and we output this factorization.
\end{enumerate}
Note that this algorithm outputs a factorization of the form  $f = \lambda.(l_1 \cdots l_n)^{\alpha}$. We can of course obtain $\lambda=1$ by an appropriate scaling of the $l_i$'s if desired.
\begin{theorem}
  The above algorithm runs in polynomial time and determines whether $f$ can be written as  $f = 
  (l_1 \cdots  l_n)^{\alpha}$ where 
  the forms $l_i$ are linearly independent. It ouputs such a factorization 
if there is one.
\end{theorem}
\begin{proof}
  The correctness of the algorithm follows from Theorem~\ref{equalexp}.
  In particular, the equivalence of properties (i) and (ii)
  in Theorem~\ref{equalexp} shows that the algorithm will make a
  correct decision on the existence of a suitable factorization at step~3.
  If this step succeeds, the simultaneous diagonalization at step~4 is possible
  since (as already pointed out Section~\ref{simulback} and
  in the proof of Theorem~\ref{equalexp})
  simultaneous diagonalization is always possible for a
  family of matrices which are diagonalizable and pairwise commute.
  By Proposition~\ref{lieconj}, the Lie algebra of $f(A^{-1}x)$
  is the space of diagonal matrices of trace 0.
  This implies that $f(A^{-1}x)$ is a constant multiple of $m$ by the
  last part of Theorem~\ref{equalexp}, and justifies the last step of
  the algorithm.
  
Let us now explain how to implement the 5  steps.
A randomized\footnote{There is no need for randomization if $f$ is given explicitly as a sum of monomials rather than by a black box (in this case we can directly solve the linear system from Lemma~\ref{lieP}).} black box algorithm for Step~1 based on Lemma~\ref{lieP} can be found in Lemma~22 of~\cite{kayal2012affine}.
Steps 2 and 3 are mostly routine  (use Proposition~\ref{diagmat} to check that
the $B_i$'s are diagonalizable).
Step 4 (simultaenous diagonalization of commuting matrices)  is also a standard linear alegbra computation. One suggestion from Section~6.1.1 of~\cite{kayal2012affine} is to diagonalize a random linear combination of the $B_i$'s (see Section~\ref{simulback} for more details).
That matrix can be diagonalized as explained in Section~\ref{diagback}.
Finally, the scaling factor $\lambda$ at step 5 can be computed by one call to the black box for $f$.
\end{proof}

\begin{remark} \label{smallfield}
We have presented the above algorithm with a view towards factorisation over $\overline{K}$, but it is readily adapted
to factorization over some intermediate field $K \subseteq \mathbb{K} \subseteq \overline{K}$. Note in particular that to decide
the existence of a factorization at step 3, we would need to check that the matrices $B_i$ are diagonalizable over $\mathbb{K}$.
As recalled in Proposition~\ref{diagmat}, this requires an algorithm that decides whether the characteristic polynomial of a matrix 
has all its roots in~$\mathbb{K}$.
{%
  In the case $\mathbb{K}= \overline{K}$, if we stop at step~3
  we obtain a purely algebraic algorithm
  for deciding the existence of a suitable
  factorization (in particular,
  we do not need to factorize univariate polynomials or diagonalize matrices).}
\end{remark}

\subsection{General case} \label{general}

In this section we describe an algorithm which takes as input a homogeneous polynomial $f$ of degree $d=\alpha_1+\cdots+\alpha_n$ in $n$ variables, determines if it can be expressed as   $f(x)=l_1(x)^{\alpha_1} \cdots l_n(x)^{\alpha_n}$ where the $l_i$'s are linearly independent forms, and  finds such a factorization if it exists. Note that the values of the exponents $\alpha_i$ are determined
by the algorithm (they are not given as input).
We assume that $\alpha_i \geq 1$ for all $i$.
The number of distinct factors is therefore equal to the number of variables of $f$.
The case where there are more factors than variables  is related to orbit closure and we do not treat it in this section.
Let us  explain briefly explain how to handle the case where some exponents $\alpha_i$ may be 0, i.e., the case where the
number $r$ of distinct factors is smaller than the number of variables. 
In this case, $f$ has only $r$ "essential variables", i.e., it is possible to make a linear (invertible) change of variables
after which $f$ depends only on $r$ variables.
This puts us therefore in the situation where the number of distinct factors is equal to the number of variables.
The number of essential variables and the corresponding change of variables can be computed with Kayal's algorithm\footnote{The algorithm in~\cite{kayal11}
  works in the circuit model, i.e., it is assumed that the input polynomial
  is given by an arithmetic circuit. Kayal later showed how to perform
  the same task in the black box model, see Section~3
  of~\cite{kayal2012affine}.}~\cite{kayal11},
see also~\cite{carlini06}.

We can now present our factorization algorithm. Like in the case of
equal exponents, the existence of a suitable factorization is decided in the
first three steps.
\begin{enumerate}
\item Compute a basis $B_1,\ldots,B_k$ of the  Lie algebra of $f$.
\item Reject if $k \neq n-1$, i.e., if the Lie algebra is not of dimension $n-1$.

\item Check that the matrices $B_1,\ldots,B_{n-1}$ commute and are 
  all diagonalizable over $\overline{K}$. 
  If this is not the case, reject.
  Otherwise, declare the existence of a factorization
  $f(x)=l_1(x)^{\alpha_1} \cdots l_n(x)^{\alpha_n}$
  where the linear forms $l_i$ are linearly independent
  and $\alpha_i \geq 1$ (the $l_i$ and $\alpha_i$ will be determined
  in the last 3 steps of the algorithm).

  \item Perform a simultaneous diagonalization of the $B_i$'s, i.e.,
  find an invertible matrix $A$ such that the $n-1$ matrices
  $AB_iA^{-1}$ are diagonal.

\item At the previous step we have found a matrix~$A$ such that $g(x)=f(A^{-1}x)$ has a Lie algebra $\mathfrak{g}_g$ which is an $(n-1)$-dimensional
  subspace of the space of diagonal matrices.
  Then we compute the orthogonal of $\mathfrak{g}_g$, i.e., 
  we find a vector $\alpha=(\alpha_1,\dots,\alpha_n)$ such $\mathfrak{g}_g$
  is the space of matrices $\diag(\lambda_1,\ldots,\lambda_n)$
  satisfying $\sum_{i=1}^n \alpha_i \lambda_i = 0$.
  We normalize $\alpha$ so that $\sum_{i=1}^n \alpha_i = d$.

\item   We must have $g(x)=\lambda.m$ where $\lambda \in K^*$ 
  and $m$ is the monomial $x_1^{\alpha_1} \cdots x_n^{\alpha_n}$ (in particular, $\alpha$ must be a vector with integral entries).
  We therefore have $f(x)=\lambda.m(Ax)$ and we output this factorization.
\end{enumerate}
Again,  this algorithm outputs a factorization of the form  $f(x) = \lambda.l_1(x)^{\alpha_1} \cdots l_n(x)^{\alpha_n}$ and we can obtain $\lambda=1$ by an appropriate scaling of the $l_i$'s.
\begin{theorem} \label{general_th}
  The above algorithm runs in polynomial time and determines whether $f$ can be written as 
  $f(x)=l_1(x)^{\alpha_1} \cdots l_n(x)^{\alpha_n}$ where 
  the forms $l_i$ are linearly independent and $\alpha_i \geq 1$ for all $i$. It ouputs such a factorization 
if there is one.
\end{theorem}
\begin{proof}
  The two main steps (finding a basis of $\mathfrak{g}_f$ and simultaneous diagonalization) can be implemented efficiently as in the case of equal exponents,
  so we'll focus on the correctness of the algorithm.

Assume first that $f$ can be written as $f(x)=L_1(x)^{\beta_1} \cdots L_n(x)^{\beta_n}$  where the $L_i$'s are linearly independent forms and $\beta_i \geq 1$ for all $i$. Then $f$ is in the orbit of the monomial $M=x_1^{\beta_1} \cdots x_n^{\beta_n}$, so $\mathfrak{g}_f$ and $\mathfrak{g}_M$ are conjugate by Proposition~\ref{lieconj}. 
By Proposition~\ref{liemon}, $\mathfrak{g}_M$  is the space of diagonal matrices $\diag(\lambda_1,\ldots,\lambda_n)$
  such that $\sum_{i=1}^n \beta_i \lambda_i =0$.
  These two facts imply that the first  4 steps of the algorithm will succeed.
  The polynomial $g(x)=f(A^{-1}x)$ defined 
at step 5 has a Lie algebra which is an $(n-1)$-dimensional
subspace of the space of diagonal matrices.
 By Proposition~\ref{moncharlie}, $g$ must therefore be a monomial.
Proposition~\ref{liemon} implies that the tuple of exponents of $g$ is correctly determined at step 5, so that 
we indeed have $g=\lambda.m$ at step 6.
Note that $m$ may differ from $M$ by a permutation of indices, and likewise the factorization output by the algorithm may differ
from  $f(x)=L_1(x)^{\beta_1} \cdots L_n(x)^{\beta_n}$ by a permutation of indices and the scaling of linear forms.

Conversely, if the 3 first steps of the algorithm succeed the $B_i$
must be simultaneously diagonalizable and it follows again from  Proposition~\ref{moncharlie} that the polynomial
$g$ defined at step 5 satisfies $g =\lambda.m$ where $\lambda \in K^*$ and $m$ is some monomial $ x_1^{\alpha_1} \cdots x_n^{\alpha_n}$. In particular,  Proposition~\ref{moncharlie} guarantees that the exponents $\alpha_i$ are all positive.
The algorithm will then output at step 6 a correct factorization of $f$.
\end{proof}
Like in Section~\ref{equal} we have presented our algorithm with a view towards factorisation over $\overline{K}$, but it is readily adapted
to factorization over some intermediate field
$K \subseteq \mathbb{K} \subseteq \overline{K}$
as explained in Remark~\ref{smallfield}.

{  In the above algorithm we need to perform the simultaneous diagonalization at
step~4 before computing the exponents $\alpha_i$. In the remainder of
this section we show that the exponents can be computed without step~4.
The corresponding algorithm relies on Proposition~\ref{eigenexp} below.
First,} we recall that for any set of matrices $S \subseteq M_n(K)$
the centralizer of
$S$ is the set of matrices that commute with all matrices of $S$.
It is a linear subspace of $M_n(K)$ and we denote it by $C(S)$.
\begin{proposition} \label{eigenexp}
  Consider a monomial $m=x_1^{\alpha_1} \cdots x_n^{\alpha_n}$
  with $\alpha_i \geq 1$ for all~$i$,
  and a polynomial $f$ in the orbit of $m$.
  
  The centralizer $C(\mathfrak{g}_f)$
  of the Lie algebra of $f$ is of dimension $n$.
  Moreover, there is a unique $H$ in $C(\mathfrak{g}_f)$ such that $\tr H =d$
  and $\tr (HM) =0$ for all $M \in \mathfrak{g}_f$.
  The matrix $H$ is diagonalizable, its eigenvalues are
  $(\alpha_1,\ldots,\alpha_n)$ and
{   $C(\mathfrak{g}_f) = \mathfrak{g_f} \oplus \mathrm{Span}(H)$.}
\end{proposition}

Note that the case 
$\alpha_1=\ldots=\alpha_n=1$
corresponds to $H=\mathrm{Id}$.
{  The condition $\tr (HM) =0$ for all $M \in \mathfrak{g}_f$
is an analogue of the trace zero condition in property (ii) of
Theorem~\ref{equalexp}.}

\begin{proof}
We first consider the case $f=m$. By Proposition~\ref{liemon},
$\mathfrak{g}_m$  is the set of diagonal matrices
$\diag(\lambda_1,\dots,\lambda_n)$ such that
$\sum_i\alpha_i\lambda_i=0$.

For $1\leq i\neq j\leq n$, let $\mathcal{H}_{ij}$ denote the set of
matrices $\diag(\lambda_1,\dots,\lambda_n)$ such that
$\lambda_i=\lambda_j$. 
Consider the set $\mathfrak{h}$ of diagonal matrices. 
The hyperplane $\mathfrak{g}_m$ of $\mathfrak{h}$ is equal to no
hyperplane of the form $\mathcal{H}_{ij}$.
Since the field is infinite, $\mathfrak{g}_m$ is not contained in the
union of the hyperplanes $\mathcal{H}_{ij}$. Then $\mathfrak{g}_m$
contains a matrix $M_0$ with pairwise distinct eigenvalues.
Then 
$\lh\subseteq C(\mathfrak{g}_m)\subseteq C(M_0)\subseteq \lh$, and
$C(\mathfrak{g}_m)=\lh$.

 Set $H_0=\diag(\alpha_1,\dots,\alpha_n)\in\lh$.
It is clear that $\tr(H_0)=d$ and $\tr(H_0M)=0$ for any $M\in
\mathfrak{g}_m$. 
Conversely, let 
$H=\diag(\beta_1,\dots,\beta_n)\in\lh$ and
$M=\diag(\lambda_1,\dots,\lambda_n)\in \mathfrak{g}_m$. 
Then $\tr(HM)=\sum_i\beta_i\lambda_i$. 
Since  $\mathfrak{g}_m$ is the hyperplane of $\lh$ defined by
$\sum_i\alpha_i\lambda_i=0$, $\tr(HM)=0$, for any $M\in
\mathfrak{g}_m$ if and only if $(\beta_1,\dots,\beta_n)$ is
proportional to $(\alpha_1,\dots,\alpha_n)$.  
If moreover $\tr(H)=d$, we get $H=H_0$. This proves the unicity.
Morever, $\tr(H_0^2)=\sum_i\alpha_i^2\neq 0$ and $H_0\not\in \mathfrak{g}_m$,
since the field has
characteristic zero. Then, since  $\mathfrak{g}_m$ is an hyperplane of
$\lh$, $\mathfrak{g}_m\oplus KH_0=C(\mathfrak{g}_m)=\lh$.\\

Consider now a point $f$ in the orbit of $m$. Let $A$ be an invertible
matrix such that $f=A.m=m\circ A^{-1}$.
Then, by Proposition~\ref{lieconj}, 
$\mathfrak{g}_f=A \mathfrak{g}_mA^{-1}$ and
$C(\mathfrak{g}_f)=A C(\mathfrak{g}_m)A^{-1}$.
One easily checks that $H$ satisfies the proposition 
for~$f$ if and
only if $A^{-1}HA$ satisfies it for $m$. With the first part, this proves the existence and unicity of $H$. 
\end{proof}
{  This proposition yields the following algorithm for the computation
of the exponents $\alpha_1,\ldots,\alpha_n$. We assume that the first three steps of the algorithm of Theorem~\ref{general_th} have executed successfully.
\begin{itemize}
\item[(a)] Set up and solve the linear system which expresses that
  $\tr[H]=d$, $\tr[HB_i]=0$ and $HB_i = B_iH$ for all $i=1,\ldots,n-1$.
  Here $(B_1,\ldots,B_{n-1})$ is the basis of $\mathfrak{g}_f$ computed
  at step 1 of the algorithm of Theorem~\ref{general_th}. The system's unknowns
  are the $n^2$ entries of $H$.

\item[(b)] Compute the eigenvalues $\alpha_1,\ldots,\alpha_n$ of $H$.
\end{itemize}
Note that the system constructed at step (a) is overdetermined: it has
$\Theta(n^3)$ equations but only $n^2$ unknowns.
Proposition~\ref{eigenexp} guarantees that the system has a unique solution
$H$, and that the eigenvalues of $H$ are the exponents
$\alpha_1,\ldots,\alpha_n$.
We refer to Section~\ref{diagback} for the computation of eigenvalues
at step~(b).}

\section{Bivariate projections} \label{bivariate}

In this section we present a probabilistic black box algorithm that finds a factorization
into products
of linear forms whenever this is possible, without any assumption of linear
independence of the linear forms.
As explained before this can be done with the algorithm by Kaltofen and Trager~\cite{KalTra90}.

We assume that the input polynomial is in $K[x_1,\ldots,x_n]$ where $K$
is infinite. In contrast to Section~\ref{factorization}, we do not need to
assume that $K$ is of characteristic 0. The hypothesis that $K$ is infinite
is needed because the algorithm draws random elements from ``large enough''
but finite subsets $S \subseteq K$. The algorithm also applies to a finite
field if $K$ is large enough for this, or if we can draw points from
a large enough field extension.

As in~\cite{kaltofen89,KalTra90} we rely on bivariate projections 
but we present a simplified algorithm
which takes advantage of the fact that we are trying to factor
polynomials of a special form 
(another simple algorithm based on a different idea is presented
in the next section).
In these two papers, a factorization of the input polynomial is recovered
from a single bivariate projection (see Step~R in~\cite{kaltofen89} and Step~1
in~\cite{KalTra90} \footnote{%
  More precisely, the construction of the black boxes for the irreducible factors of $f$ requires a single projection.
Evaluating these black boxes at an input point requires another bivariate projection, see Step~A in~\cite{KalTra90}}). By contrast, we will recover 
the solution to our problem from
several projections as in e.g.~\cite{GKP18,kayal2012affine}.
A recurring difficulty with projection-based algorithms is that when
we try to ``lift'' the solutions of problems on a lower-dimensional space
to a solution of the original problem, the lift may not be unique.
We first present in Section~\ref{unique} a solution under an additional
assumption which guarantees uniqueness of the lift. We then lift (as it were)
this assumption in Section~\ref{randomproj}.

We assume that a polynomial time factorization algorithm for polynomials in $K[x,y]$
is available. It is explained in~\cite{kaltofen85} how to obtain
such an algorithm from a univariate factorization algorithm for the field of rational numbers, and more generally for number fields and finite fields.
In the case of absolute factorization, polynomial time algoritms were first given by Gao~\cite{gao03} and by 
Ch\`eze and Lecerf~\cite{ChezeLecerf07}.
The complexity of the latter algorithm was analyzed in~\cite{ChezeLecerf07} for the algebraic (unit cost) model of 
computation. The complexity of the former algorithm  was analyzed in~\cite{gao03} for an input polynomial with coefficients
in a finite field $F_q$.\footnote{The algorithm also works for fields of characteristic 0, but a precise analysis of its complexity was left for future work.}

Without loss of generality, we'll assume that our input $f$ is
in $K[x_1,\ldots,x_n]$ with $n \geq 4$.
Indeed, if there are only 3 variables we can set
$g(x_1,x_2) = f(x_1,x_2,1)$, use the bivariate algorithm
to factor $g$ as a product of affine forms, and homogonenize the result to
obtain a factorization of $f$. Note that the homogenization step includes
a multiplication by $x_3^{\deg(f)-\deg(g)}$.

\subsection{A uniqueness condition} \label{unique}

In this section we assume that our input $f(x_1,\ldots,x_n)$
can be factorized as
\begin{equation} \label{manyforms}
  f(x)=\lambda.l_1(x)^{\alpha_1} \cdots  l_k(x)^{\alpha_k}
  \end{equation}
where the linear form $l_i$ is not proportional to $l_j$ if $i \neq j$,
and $\lambda$ is a nonzero constant.
We would like to recover $\lambda$, the exponents $\alpha_i$'s and the $l_i$'s
(note that each linear forms is defined only up to a constant).

Write $l_i(x) =\sum_{j=1}^n l_{ij}x_j$.
In order to guarantee ``uniqueness of the lift'' we make
the following temporary assumption:
\begin{itemize}
\item[(*)] The $k$ coefficients $l_{i1}$ are distinct and nonzero
  and $l_{in}=1$ for all $i$.
 \end{itemize}
The algorithm is as follows.
\begin{enumerate}
\item For $j=2,\ldots,n-1$ define $g_j(x_1,x_j) = f \circ \pi_j$
  where the projection $\pi_j$ sends variable $x_n$ to the constant 1,
  leaves $x_1$ and $x_j$ unchanged and sets all other variables to 0.
  Compute the dense representation of the $g_j$'s by interpolation.

\item Using the bivariate factorization algorithm, write each $g_j$
  as $g_j(x_1,x_j)=\lambda. \prod_{i=1}^k (a_{ij}x_1+b_{ij}x_j+1)^{\beta_{ij}}$.

\item At the beginning of this step, each of the $n-2$ tuples
  $(a_{1j},\ldots,a_{kj})$ is a permutation of the tuple $(l_{11},\ldots,l_{k1})$.
  We reorder the factors in the factorizations of the $g_j$
  from step 2 to make sure that the $n-2$ tuples are identical (i.e.,
  its elements always appear in the same order).
  After reordering, the $n-2$ tuples of exponents
  $(\beta_{1j},\ldots,\beta_{kj})$ will also become identical.
  We therefore obtain factorizations of the form:
  $$g_j(x_1,x_j)=\lambda. \prod_{i=1}^k (a_ix_1+c_{ij}x_j+1)^{\gamma_i}.$$

\item We output the factorization:
  $$f(x_1,\ldots,x_n)=\lambda.
  \prod_{i=1}^k (a_ix_1+c_{i2}x_2+\cdots+c_{i,n-1}x_{n-1}+x_n)^{\gamma_i}.$$
\end{enumerate}
The main issue regarding the correctness of this algorithm is to make sure
that we have correctly combined the factors of the $g_j$'s to obtain
the factors of $f$.
This is established in the next proposition.
For an example of what can go wrong without assumption (*) consider the
following two polynomials:
$$f_1=(x_1+x_2+x_3+x_4)(x_1+2x_2+2x_3+x_4)$$
and
$$f_2=(x_1+x_2+2x_3+x_4)(x_1+2x_2+x_3+x_4).$$
At step 1 of the algorithm, these two polynomials are mapped to the same
pair of bivariate polynomials:
$$g_2=(x_1+x_2+1)(x_1+2x_2+1),\ g_3=(x_1+x_3+1)(x_1+2x_3+1)$$
and there is no unique way of lifting $\{g_2,g_3\}$ to an input polynomial.
Another difficulty is that the factorization pattern of $f$ (i.e., the set of exponents $\{\alpha_1,\ldots,\alpha_k\}$) could change after projection, for instance $$f=(x_1+x_2+x_3+x_4)(x_1+2x_2+x_3+x_4)$$ is mapped to
$$g_2=(x_1+x_2+1)(x_1+2x_2+1),\ g_3=(x_1+x_3+1)^2.$$

\begin{proposition} \label{assumptionprop}
  The above algorithm correctly factorizes the polynomials of form~(\ref{manyforms}) that satisfy assumption (*).
\end{proposition}
\begin{proof}
  Since $l_{in}=1$ for all $i$ we have $\lambda=f(0,\cdots,0,1)=g_j(0,\cdots,0)$
  for all $j=2,\ldots,n-1$.
  Each $g_j$ admits the factorization:
  \begin{equation} \label{2factor}
    g_j(x_1,x_j)=\lambda. \prod_{i=1}^k (l_{i1}x_1+l_{ij}x_j+1)^{\alpha_j}
    \end{equation}
  All these polynomials therefore have same factorization pattern as $f$
  (note in particular that the affine forms $l_{i1}x_1+l_{ij}x_j+1$ are nonconstant since $l_{i1} \neq 0$; and two of these forms cannot be proportional
  since the $l_{i1}$ are distinct).  It follows that the factorization of $g_j$ discovered by the algorithm
  at step 2 is identical to~(\ref{2factor}) up to a permutation,
  i.e., we have $a_{ij}=l_{\sigma_j(i)1}$, $b_{ij}=l_{\sigma_j(i)j}$
  and $\beta_{ij}=\alpha_{\sigma_j(i)}$ for some permutation
  $\sigma_j \in {\mathfrak{S}}_k$.
  Since the $l_{i1}$ are distinct, after reordering at step 3 these $n-2$ permutations become identical, i.e.,
  we have $a_{i}=l_{\sigma(i)1}$, $c_{ij}=l_{\sigma(i)j}$
  and $\gamma_{i}=\alpha_{\sigma(i)}$ for some permutation~$\sigma$.
  Finally, at step 4 the algorithm outputs the correct factorization
  $f(x)=\lambda.\prod_{i=1}^k l_{\sigma(i)}(x)^{\alpha_{\sigma(i)}}$.
 \end{proof}

\subsection{General case} \label{randomproj}

In this section we present a black box algorithm that factors a homogeneous
polynomial $f \in K[x_1,\ldots,x_n]$ of degree $d$ into a product of $d$
linear forms whenever this is possible, thereby lifting assumption (*)
from Section~\ref{unique}.
The algorithm is as follows.
\begin{enumerate}
\item Set $g(x)=f(A.x)$ where $A \in M_n(K)$ is a random matrix.
\item Attempt to factor $g$ with the algorithm of Section~\ref{unique}.
  If this fails, reject $f$. In case of success,
  let $g'(x)=\lambda.l_1(x)^{\alpha_1} \cdots  l_k(x)^{\alpha_k}$
  be the factorization output by this algorithm.
\item Check that $f(x)=g'(A^{-1}.x)$ and output the corresponding factorization.
\end{enumerate}
The random matrix at step~1 is constructed by drawing its entries
independently at random from some large enough finite set $S \subseteq K$.
The point of this random change of variables is that $g$ will satisfy
assumption (*) of Section~\ref{unique} with high probability if $f$
can be factored as a product of linear forms. The (quite standard) arguments
needed to estimate the probability of success are presented in the
proof of Theorem~\ref{projth}.
Note also that by the Schwarz-Zippel lemma, $A$ will be invertible with high probability.

At step 2 we need a black-box for $g$. Such a black box is easily
obtained by composing the black box for $f$ with the map $x \mapsto A.x$.

At step 3, we check the polynomial identity $f(x)=g'(A^{-1}.x)$
by evaluating the left and right-hand sides at 
{one random point}.
\begin{theorem} \label{projth}
  The above algorithm runs in polynomial time and determines whether $f$
  can be written as a product of linear forms.
 It outputs such a factorization if there is one.
\end{theorem}
\begin{proof}
  By the Schwarz-Zippel lemma, any factorization of $f$ output at step 3
  will be correct with high probability.
  So we only need to prove the converse: if $f$ can be factored as a product
  of linear forms, the algorithm finds a correct factorization with high
  probability. Suppose therefore that 
$$f(x)=L_1(x)^{\alpha_1} \cdots  L_k(x)^{\alpha_k}$$
  where no two linear forms $L_i, L_j$ in this expression are proportional.
  Then $g(x)=f(A.x)$ can be written as
  $$g(x)=\ell_1(x)^{\alpha_1} \cdots  \ell_k(x)^{\alpha_k}$$
  where $\ell_i(x)=L_i(A.x)$.
  If $A$ is invertible, the linear forms in this expression will not
  be proportional.
  The coefficients of these linear forms are given by the expression:
  \begin{equation} \label{elleq}
    \ell_{ij}=\sum_{p=1}^n L_{ip} A_{pj}.
    \end{equation}
  If the entries $A_{pj}$ of $A$ are drawn from a set $S \subset K$,
  $\ell_{in}=0$ with probability at most $1/|S|$
  since $L_i {\not \equiv} 0$. These $n$ coefficients will all be nonzero
  with probability at least $1-n/|S|$; in this case we can
  factor out $\lambda=\prod_{i=1}^k \ell_{in}^{\alpha_i}$ to make sure that
  the coefficient of $x_n$ in each linear form is equal to 1 as required
  by assumption (*).
  This gives the factorization
  $$g(x)=\lambda.l_1(x)^{\alpha_1} \cdots l_k(x)^{\alpha_k}$$
  where $l_i(x)=\ell_i(x)/\ell_{in}$.
  The same argument as for $\ell_{in}$ shows that $\ell_{i1}$ and therefore
  $l_{i1}$ will be nonzero with high probability.
  To take care of assumption~(*), it remains to check that the $l_{i1}$
  will be  distinct with high probability.
  The condition $l_{i1} \neq l_{j1}$ is equivalent to
  $\ell_{i1}\ell_{jn} - \ell_{j1} \ell_{in} \neq 0$.
  By~(\ref{elleq}) this expression can be viewed as a quadratic form
  in the entries of $A$. From unique factorization and the
  hypothesis that the linear forms $L_i, L_j$ are not proportional it follows
  that this quadratic form is not identically 0. We conclude again
  that it will be nonzero with high probability by the Schwarz-Zippel lemma.

  We have established that $g(x)=f(A.x)$ satisfies (*) with high probability.
  In this case, by Proposition~\ref{assumptionprop}
  the factorization of $g$ at step 2 of the algorithm 
  and the verification of the polynomial identity at step 3 will also
  succeed.
\end{proof}

\section{Identifying the hyperplanes and their multiplicities}
\label{hyperplane}

If a polynomial $f$ can be factored as a product of linear forms, its zero
set $Z(f)$
is a union of (homogeneous) hyperplanes.
In this section we present an algorithm based on this simple geometric
fact.

We can identify each hyperplane in $Z(f)$
by finding $n-1$ nonzero points that lie on it. Assume that
$f$ can be written as
$f(x)=\lambda.l_1(x)^{\alpha_1} \cdots l_k(x)^{\alpha_k}$ where the linear forms $l_i$ are
not proportional. We will need a total of $k(n-1)$
points on $Z(f)$ to identify the $k$ hyperplanes.
Our algorithm begins with the determination of these $k(n-1)$ points.
\begin{enumerate}
\item Pick a random point $a \in K^n$ and $n-1$ random
  vectors $v_1,\ldots,v_{n-1}$ in~$K^n$  
(or representatives of points in $\mathbb{P}(K^{n})$ to be more precise). 

  Let $\Delta_i$ be the line of direction $v_i$ going through $a$.
  Compute the intersection $\Delta_i \cap Z(f)$ for $i=1,\ldots,n-1$.

\item Output the $k(n-1)$ intersection points $a_1,\ldots,a_{k(n-1)}$ found at step~1.
\end{enumerate}
In the sequel, we assume that $f(a) \neq 0$. This holds with high probability
by the Schwarz-Zippel lemma.

At step~1 we compute $\Delta_i \cap Z(f)$ by finding the roots of
the univariate polynomial $g(t)=f(a+tv_i)$. We obtain one point
on each hyperplane $Z(l_1),\ldots,Z(l_k)$ except if $v_i$ belongs to
one of these hyperplanes. This can happen only with negligible probability.
Moreover, these $k$ points are distinct except if $\Delta_i$ goes
through the intersection of two of these hyperplanes.
Again, this happens with negligible probability (we explain in the proof of
Theorem~\ref{hyperplaneth} how to obtain explicit bounds on the probabilities
of these bad events).
Since $a {\not \in Z(f)}$, with high probability we find a total of
$k(n-1)$ distinct points as claimed at step~2.
Moreover, each hyperplane $Z(l_i)$ contains exactly $n-1$ points.
Note that at step~1 we have also determined $k$ if this parameter was
not already known in advance.

At the next stage of our algorithm we determine the $k$ hyperplanes.
We first determine the hyperplane going through $a_1$ as follows:
\begin{enumerate}
\item[3.] Find $n-2$ points $b_2,\ldots,b_{n-1}$ in the set
  $\{a_2,\ldots,a_{k(n-1)}\}$ such that each line $(a_1 b_j)$ is included
  in $Z(f)$.

\item[4.] Output the linear subspace
  $H_1=\mathrm{Span}(a_1,b_2,\ldots,b_{n-1})$.
\end{enumerate}
At step~3 we can find out whether a line $(a_1a_j)$ is included in $Z(f)$
by checking that the univariate polynomial $g(t)=f(ta_1+(1-t)a_j)$
is identically~0.
This can be done deterministically with $k-1$ calls to the black box for $f$
(indeed, if $g {\not \equiv} 0$ this polynomial has at most $k$ roots,
and we already know that $g(0)=g(1)=0$). Alternatively, we can perform
a single call to the black box by evaluating $g$ at a random point.

Assume for instance that $Z(l_1)$ is the hyperplane going through
$a_1$. In the analysis of the first two steps we saw that (with high probability)
$a_1$ does not lie on any other $Z(l_j)$, and that exactly $n-2$ points
$b_2,\ldots,b_{n-1}$ in $\{a_2,\ldots,a_{k(n-1)}\}$ lie on $Z(l_1)$.
The algorithm identifies these points at step 3
(we will find exactly one point on each line $\Delta_i$).
It follows that the subspace $H_1$ output at step 4 is included in $Z(l_1)$.
To conclude that $H_1 = Z(l_1)$, it remains to show that $H_1$ is 
of dimension~$n-1$.
Assume without loss of generality that $\{a_1\}=\Delta_1 \cap Z(l_1)$
and $\{b_j\} = \Delta_j \cap Z(l_1)$ for $j=2,\ldots,n-1$.
Then $a_1=a+t_1v_1$ and $b_j=a+t_jv_j$ for $j=2,\ldots,n-1$.
Here $v_1,\ldots,v_{n-1}$ are the directions chosen at step 1,
and $t_1,\ldots,t_{n-1}$ are appropriate nonzero scalars.
With high probability, the $n$ vectors $a,v_1,\ldots,v_{n-1}$ are linearly
independent. In this case, the family $a+t_1v_1,\ldots,a+t_{n-1}v_{n-1}$
is of rank $n-1$ as desired.

The above analysis shows that steps 3 and 4 identify $H_1=Z(l_1)$
with high probability. The $k-1$ remaining hyperplanes can
be identified by repeating this procedure.
For instance, to determine the second hyperplane $H_2$ we will remove
the points $a_1,b_2,\ldots,b_{n-1}$ (which lie on $H_1$)
from the set $\{a_1,\ldots,a_{k(n-1)}\}$ and we will determine the hyperplane
going through the first of the $(k-1)(n-1)$ remaining points.

In the next stage of the algorithm we determine the multiplicities
$\alpha_i$ of the linear forms $l_i$. This is done as follows:
\begin{enumerate}
\item[5.] Consider again the random point $a$ and the random vector $v_1$
  drawn at step 1. We have already computed the intersection points
  with $H_1=Z(l_1),\ldots,H_k=Z(l_k)$ of the line $\Delta_1$ of direction
  $v_1$ going through~$a$.
  Recall that this was done by computing the roots $t_1,\ldots,t_k$
  of the univariate polynomial $g(t)=f(a+tv_1)$.

  Let us assume without loss of generality that these roots are ordered
  so that $\{a+t_1v_1\} = H_1 \cap \Delta_1,\ldots, \{a+t_k v_1 \} = H_k \cap \Delta_1$.
  Now we compute the multiplicities $\alpha_1,\ldots,\alpha_k$
  of $t_1,\ldots,t_k$ as roots of $g$ and we output these multiplicities.
\end{enumerate}
If $f(x)=l_1(x)^{\alpha_1} \cdots l_k(x)^{\alpha_k}$, the multiplicities
of the roots of $g$ are indeed equal to $\alpha_1,\ldots,\alpha_k$ except if
$\Delta_1$ goes through the intersection of two of the hyperplanes $H_1,\ldots,H_k$. As already pointed out  in the analysis of the first two steps,
this happens only
with negligible probability.
Note that there is nothing special about $\Delta_1$ at step 5: we could
have used a new random line~$\Delta$ instead.

The final stage of the algorithm is a normalization step.
\begin{enumerate}
\item[6.] At the beginning of this step we have determined linear forms $l_i$
  and multiplicities $\alpha_i$
  so that $f(x)=\lambda.l_1(x)^{\alpha_1} \cdots l_k(x)^{\alpha_k}$
  for some constant~$\lambda$. We determine~$\lambda$ by one call to the black box for $f$
  at a point where the $l_i$ do not vanish (for instance, at a random point).
\end{enumerate}
We have obtained the following result.
\begin{theorem} \label{hyperplaneth}
  Let $f \in K[x_1,\ldots,x_n]$ be a polynomial of degree $d$
  that admits a factorization
  $f(x)=\lambda.l_1(x)^{\alpha_1} \cdots l_k(x)^{\alpha_k}$ over $\overline{K}$,
  where no two linear forms $l_i$ are proportional.
  The above algorithm determines such a factorization with high probability,
  and the number of calls to the black box for $f$ is polynomial
  in $n$ and $d$. 

  Assume moreover that $K=\qq$ and that a factorization of $f$
  where $l_i \in \qq[x_1,\ldots,x_n]$ is possible.
  If the coefficients of these linear forms are of bit size
  at most $s$ then all calls to the black box are made at rational points of
  bit size polynomial in $n$, $d$ and $s$.
\end{theorem}
\begin{proof}
The correctness of the algorithm follows from the above analysis.
Let us focus therefore on the case $K=\qq$ of the theorem.
This result relies on a standard application of the Schwarz-Zippel
lemma.
More precisely,
as explained in the analysis of the first two steps, we want to pick
a random point $a$ such that $f(a) \neq 0$ and random vectors
$v_1,\ldots,v_{n-1}$ that do not belong to any of the hyperplanes.
Moreover, the line $\Delta_i$
defined at step~1 should not go through the intersection of two hyperplanes.
Let us pick the coordinates of $a$ and of the $v_i$ independently at random
from a finite set $S \subseteq \qq$. By the Schwarz-Zippel lemma,
$\Pr[f(a)=0] \leq k/|S| \leq d/|S|$;
and for any linear form $l_j$
we have $\Pr[l_j(v_i)=0] \leq 1/|S|$.
As to $\Delta_i$, let us bound for instance the probability of going
through the intersection of the first two hyperplanes.
Since $\Delta_i$ is the line of direction $v_i$ going through~$a$,
it suffices to make sure that $l_2(a)l_1(v_i) - l_1(a)l_2(v_i) \neq 0$.
By the Schwarz-Zippel lemma this happens with probability at least
$1-2/|S|$.

Another constraint arising in the analysis of steps 3 and 4 is that $a,v_1,\ldots,v_{n-1}$ should be linearly independent. By the Schwarz-Zippel lemma, the
corresponding determinant vanishes with probability at most~$n/|S|$.

Note that the bounds obtained so far are independent of $s$.
This parameter comes into play when we compute the intersections
$\Delta_i \cap Z(f)$ at step~1. Recall that we do this by
finding the roots of the univariate polynomial $g(t)=f(a+tv_i)$.
The roots are:
$$t_1=-l_1(a)/l_1(v_i),\ldots,t_k=-l_k(a)/l_k(v_i).$$
Then at step 3 we call the black box at points belonging to lines
going through two of the $k(n-1)$ intersections points found at step~1.
\end{proof}

The algorithm presented in this section relies on a simple and appealing
geometric picture, but it suffers from a drawback compared to the algorithms
of sections~\ref{factorization} and~\ref{bivariate}:
\begin{remark} \label{pointsize}
  Assume that $K = \qq$. The above algorithm may need to call the black box
  for $f$ at algebraic (non rational) points in the case where the linear forms
  $l_i$ do not have rational coefficients.
  This is due to the fact that we call the black box at points
  that lie on the hyperplanes $l_i=0$.

  By contrast, the algorithms of sections~\ref{factorization}
  and~\ref{bivariate} always call the black box at integer points
  even when $f$ has algebraic (non rational) coefficients.
  To see why this is true for the algorithm of Section~\ref{bivariate},
  note that the main use of the black box is for performing bivariate
  interpolation. In Section~\ref{factorization}, the black box
  is used
  {%
    only} for the computation of the Lie algebra of $f$
  following Lemma~22 of~\cite{kayal2012affine}.
  {%
    More details on the black box calls performed by our
    three algorithms can be found in the appendix.}
\end{remark}

\appendix

\normalsize

\section{Appendix: Cost of calls to the black box} \label{bb}

In this section we compare the number of calls to the black box
made by our three algorithms { (for the cost of other operations, see Appendix~\ref{other})}. A more thorough analysis would also
take into account the size of points at which the black box is queried
(for this, Remark~\ref{pointsize} would become especially relevant).

It turns out that the hyperplane algorithm of Section~\ref{hyperplane}
makes fewer calls to the black box than the other two.
We also analyze these algorithms
in the ``white box'' model, where we have access to an arithmetic
circuit computing the input polynomial $f$. In that model, the
cost of function evaluations becomes smallest for the 
Lie-theoretic algorithm of Section~\ref{factorization}.

\subsection{Lie-theoretic algorithm} \label{bb_lie}

In Section~\ref{factorization}, the black box
  is used   only  for the computation of the Lie algebra of $f$.
  By Lemma~\ref{lieP}, this boils down to the determination of linear
  dependence relations between the $n^2$ polynomials
  $x_j\frac{\partial f}{\partial x_i}.$
  The general problem of finding linear dependence relations between
  polynomials given by black box access is solved by the following lemma
  (see appendix~A1 of~\cite{kayal11} for a proof).
 \begin{lemma}[Lemma~14 in~\cite{kayal2012affine}]
   Let $(f_1(x),f_2(x),\ldots,f_m(x))$ be an $m$-tuple
  of $n$-variate polynomials. Let ${\cal P} = \{a_i;\ 1 \leq i \leq m\}$
  be a set of $m$ points in $K^n$. Consider the $m \times m$ matrix
  $$M=(f_j(a_i))_{1 \leq i,j \leq m}.$$
  With high probability over a random choice of ${\cal P}$, the nullspace
  of $M$ consists precisely of all the vectors
  $(\alpha_1,\ldots,\alpha_m) \in K^m$ such that
  $$\sum_{i=1}^m \alpha_i f_i(x) \equiv 0.$$
  \end{lemma}
  We therefore need to evaluate the $n$ polynomials
  $\partial f / \partial x_i$ at $n^2$ random points.
  Note however that we have only access to a black box for $f$ rather than
  for its partial derivatives. As is well known, it is easy to
  take care of this issue by polynomial interpolation. Suppose indeed
  that we wish to evaluate $\partial f / \partial x_i$ at a point $a$.
  Then we evaluate $f$ at $d+1=\deg(f)+1$ points on the line $\Delta$ which
  goes through $a$ and is parallel to the $i$-th basis vector.
  From these $d+1$ values we can recover $f$ and $\partial f / \partial x_i$
  on $\Delta$.
  We conclude that the Lie-theoretic algorithm performs $O(dn^3)$ calls
  to the black box for $f$.

  {%
    These $n^3$ polynomial interpolations also have a cost in terms
    of arithmetic operations, but it is relatively small. Suppose indeed
    that we wish to compute ${\partial f} / {\partial x_1}$ at a point
    $a=(a_1,\ldots,a_n)$ with $a_1 \neq 0$. Consider the univariate
    polynomial $g(x)=f(a_1x,a_2,\ldots,a_n)$. It suffices to compute
    $g'(1)=a_1 {\partial f} / {\partial x_1}(a)$. We can obtain
    $g'(1)$ as a fixed linear combination of $g(0),g(1),\ldots,g(d)$.
    One polynomial interpolation therefore requires one linear combination
    of values of $f$ and one division (by $a_1$).
    {
      We conclude for use in Appendix~\ref{other_lie} that the arithmetic
      cost of these $n^3$ interpolations is $O(n^3d)$.}

  The Lie-theoretic algorithm
  admits an interesting optimization when the black box
  is implemented by an arithmetic circuit. Suppose indeed that
  we have access to an arithmetic circuit of size $s$ computing $f$
  (this is the so-called ``white box'' model). The above analysis
  translates immediately into an arithmetic cost of order $sdn^3$ for
  the evaluation of the partial derivatives of $f$ at our $n^2$ random points.
  But one can do much better thanks to the classical result
  by Baur and Strassen~\cite{baur83} (see also~\cite{morgenstern85}),
  which shows that the $n$ partial derivatives of an arithmetic circuit
  of size $s$ can be evaluated by a single arithmetic circuit
  of size $O(s)$. This reduces the cost of evaluations from $O(sdn^3)$ to
  $O(sn^2)$.
  {
    Moreover, the arithmetic cost of interpolations drops from $O(n^3d)$ to 0
    since we do not perform any interpolation in the white box model.}

  \subsection{Bivariate projections} \label{bb_bivariate}

  The algorithm of Section~\ref{bivariate} recovers a factorization of
  the input polynomial $f$ from the factorization of $n-2$ bivariate
  projections $g_1,\ldots,g_{n-2}$. The black box is used only to obtain
  each $g_j$ in dense form by interpolation.\footnote{There is also an additional call to the black box for verification of the final result, see Step~3 in Section~\ref{randomproj}.} This can be done deterministically
  by evaluating $g_j$ on any set of the form $S \times S$ where $|S| = d+1$.
  We therefore need a total of  $O(nd^2)$ function evaluations.
  Note that this is only $O(n^3)$ for $d=n$, i.e., smaller than the
  number of black box call performed by the Lie-theoretic algorithm.
  In general the bounds $dn^3$ and $nd^2$
  obtained for our first two algorithms are not
  comparable since $d$ could be  much larger than $n$ (this can happen
  when the exponents $\alpha_i$ in~(\ref{problem}) are large enough).

  There is no obvious improvement to this analysis of our second algorithm
  in the ``white box'' model described in Section~\ref{bb_lie}:
  the  $O(nd^2)$ function evaluations
  translate into an arithmetic cost of order $snd^2$.
    Now the white box version of the Lie-theoretic algorithm
  becomes more interesting from the point of view
  of the cost of function evaluations: as explained above, this
  cost is only $O(sn^2)$. By contrast, this cost is $\Omega(sn^3)$ for
  the algorithm of Section~\ref{bivariate} since $d \geq n$.

  \subsection{The hyperplane algorithm} \label{bb_hyperplane}

In order to factor an input polynomial
$f(x)=\lambda.l_1(x)^{\alpha_1} \cdots l_k(x)^{\alpha_k}$,
this algorithm determines the hyperplanes $H_i=Z(l_i)$ together with their
multiplicities $\alpha_i$.
The black box calls are performed at steps 1, 3 and 6 of the algorithm.
We'll focus on steps 1 and 3 since Step 6 performs only one call to
the black box.

At Step~1 we compute the intersection of $n-1$ lines
$\Delta_1,\ldots,\Delta_{n-1}$ with $Z(f)$, the zero set of $f$.
For this we need to interpolate $f$ on each line; this requires a total
of $(n-1)(d+1)$ calls to the black box.

The determination of a single hyperplane of $Z(f)$
is explained at step~3 of the algorithm, which we repeat here for convenience
($\{a_2,\ldots,a_{k(n-1)}\}$ are the intersection points found at Step~1):
\begin{enumerate}
\item[3.] Find $n-2$ points $b_2,\ldots,b_{n-1}$ in the set
  $\{a_2,\ldots,a_{k(n-1)}\}$ such that each line $(a_1 b_j)$ is included
  in $Z(f)$.
\end{enumerate}
As explained in Section~\ref{hyperplane}, the test $(a_1 b_j) \subseteq Z(f)$
can be implemented with one call to the black box at a random point
on the line $(a_1 b_j)$.
This test is repeated at most $k(n-1)$ times. We therefore need $O(kn)$ calls
to determine a single hyperplane. There are $k$ hyperplanes to determine,
for a total cost of order $k^2n$. 
We conclude that this algorithm makes $O(dn+k^2n)$ calls to the black box.
The two terms $dn$ and $k^2n$ in this bound are in general incomparable
since $d \geq k$
is the only relation between $d=\deg(f)$ and the number $k$ of distinct
factors.

In order to compare with the Lie-theoretic algorithm we whould
set $k=n$ since that algorithm applies only in this situation.
The cost of the hyperplane algorithm becomes $O(dn+n^3)$;
this is smaller than the $O(dn^3)$ bound obtained for the Lie-theoretic
algorithm. Note however that the latter algorithm becomes cheaper in the white
box model: as explained in Section~\ref{bb_lie} the arithmetic cost of function
evaluations is only $O(sn^2)$ when $f$ is given by an arithmetic circuit
of size $s$.
This should be compared to a cost of order $s(dn+n^3)$ for the hyperplane
algorithm (like the bivariate algorithm, it does not seem to admit
any interesting optimization in the white box model).

Finally, the hyperplane algorithm should be compared to bivariate
projections. In number of calls to the black box,
the latter algorithm is always as expensive or more
expensive than the former (compare $dn+k^2n$ to $d^2n$).

\section{Appendix: Cost of other operations} \label{other}

{
  In this section we continue the analysis of our three algorithms.
  Appendix~\ref{bb} dealt with the number of calls to the black
  box. Here we estimate the cost of ``other operations'', which
  consist mostly of:
  \begin{itemize}
  \item arithmetic operations, in $K$ or in an extension of $K$.
  \item certain non-algebraic steps such as eigenvalue computations
    or the factorization of univariate polynomials.
  \end{itemize}
  The bounds that we give should only be viewed as very rough estimates
  of the algorithms' complexity since we do not perform a full analysis
  at the level of bit operations.\footnote{ Note that a complexity analysis at the level of bit operations is also omitted from the paper by Kalftofen and Trager~\cite{KalTra90} on black box factorization.}}}

\subsection{Lie-theoretic algorithm} \label{other_lie}

In this section we analyze more precisely the algorithm of
Section~\ref{general}. We'll focus first on the complexity
of deciding the existence of a suitable factorization over $\overline{K}$.
This is done in the first three steps of the algorithm.
Note that the corresponding steps for the case of equal exponents
(Section~\ref{equal}) only differ by the presence of $n-1$ trace
computations. The cost of trace computations turns out to be negligible,
so this analysis applies to the two versions of our algorithm.

At Step 1 of the algorithm we compute a basis of the Lie algebra of $f$.
The Lie algebra is the nullspace of a certain matrix $M$ of size $m=n^2$
which we have already computed as explained in Appendix~\ref{bb_lie}.
A basis of the nullspace can be computed with $O(m^3)$ arithmetic operations
by Gaussian elimination, and with $O(m^{\theta})$ operations
using fast linear algebra~\cite{BCS}.
Here $\theta$ denotes any exponent strictly larger than $\omega$,
the exponent of matrix multiplication.

At Step 3 we first check that the matrices $B_1,\ldots,B_{n-1}$ commute,
where $B_1,\ldots,B_{n-1}$ is the basis of the Lie algebra found at Step~1.
This can be done in $O(n^{2+\omega})$ arithmetic operations.
This is negligible compared to the cost $O(n^{2\theta})$ of the first step
since $\theta > \omega \geq 2$.

Then we check that the $B_i$ are all diagonalizable.
Recall from Section~\ref{diagback} that $B_i$ is diagonalizable over
$\overline{K}$ iff its minimal polynomial $m_i$ has only simple roots.
The minimal polynomial can be computed in $O(n^{\theta})$ arithmetic
operations~\cite{giesbrecht95,storjohann01}. Then we need to check
that $\mathrm{gcd}(m_i,m'_i)=1$. The cost of computing the gcd is
negligible compared to $n^{\theta}$.
The cost of the $n-1$ diagonalizability tests
is $O(n^{1+\theta})$, which is again negligible compared to Step~1.
We conclude that the existence of a suitable factorization of $f$ can
be decided in $O(n^{2\theta})$ arithmetic operations.

At Step~4 we perform a simultaneous diagonalization of the $B_i$.
As suggested in
Section~\ref{simulback} and in Section~\ref{factorization},
this can be done by diagonalizing a random combination $R$ of the $B_i$'s.
For this, as recalled in Section~\ref{diagback}
we can first compute the eigenvalues $\lambda_1,\ldots,\lambda_n$
of $R$ (a non algebraic step). Then we compute a basis of $\ker(R-\lambda_iI)$
for all $i$. One basis can computed in $O(n^{\theta})$ arithmetic operations,
so $O(n^{1+\theta})$ is a rough estimate on the number of arithmetic operations
needed to compute a transition matrix $T$ (we will not try to improve it
since it is dominated by the cost of Step 1).
Note that these arithmetic operations take place in
$K[\lambda_1,\ldots,\lambda_n]$, so counting such an operation as ``one step''
is probably most appropriate when the $\lambda_i$ lie in $K$, or when we
work with approximations of the $\lambda_i$.

Once $T$ is known, the $n-1$ diagonal matrices $D_i=T^{-1}B_iT$ can be computed
at a cost of $O(n^{1+\theta})$ arithmetic operations. Then, as explained
at Step 5, we obtain the exponents $\alpha_i$ as the orthogonal of the
space spanned by the $D_i$ in the space of diagonal matrices.
Alternatively, the $\alpha_i$ can be obtained without knowledge of $T$
as explained after Proposition~\ref{eigenexp}: the $n$ exponents
are the eigenvalues of matrix $H$ which is obtained as the unique
solution of a system of $\Theta(n^3)$ equations in $n^2$ unknowns.
This approach looks rather expensive since solving a {\em square} system
in $n^2$ unknowns would already take $O(n^{2\theta})$ arithmetic operations.

The above analysis can be summarized as follows.
\begin{proposition} \label{cost_lie}
  The algorithm of Section~\ref{general} decides in { $O(n^{2\theta}+n^3d)$}
  arithmetic operations whether the input polynomial $f$ admits
  a factorization of the form~(\ref{problem}) over $\overline{K}$.
  If such a factorization exists, it can be computed within the same
  number of arithmetic operations and the additional computation
  of the eigenvalues of a matrix $R \in M_n(K)$.

  {
    In the white box model of Appendix~\ref{bb_lie}, the number of
    arithmetic operations drops
    from $O(n^{2\theta}+n^3d)$ to $O(n^{2\theta})$.}
\end{proposition}
{
  The term $n^3d$ in Proposition~\ref{cost_lie} is due to the arithmetic
  cost of interpolations as explained in Appendix~\ref{bb_lie}.}
Towards a more thorough analysis of this algorithm one could attempt
to estimate its bit complexity, assuming for instance for simplicity
that $f$ admits
a factorization with the $l_i$ in $\zz[x_1,\ldots,x_n]$.

\subsection{Bivariate projections} \label{other_bi}

The algorithm from Section~\ref{bivariate} recovers a factorization
of $f$ from $n-2$ bivariate factorization.
A state of the art algorithm for the latter task
can be found in~\cite{lecerf10}, where the following reduction
from bivariate to univariate factorization is provided.
\begin{theorem} \label{lecerf}
  Let $\mathbb{K}$ be a field of characteristic 0, and $F \in \mathbb{K}[x,y]$
  a bivariate polynomial of degree $d_x$ in the variable $x$ and $d_y$ in
  the variable $y$.
  There is a probabilistic algorithm that factors $F$ in $O((d_xd_y)^{1.5})$
  arithmetic operations.
  Moreover, the algorithm performs irreducible factorizations of polynomials
  in $\mathbb{K}[y]$ whose degree sum is at most $d_x+d_y$.
\end{theorem}
We have omitted the cost of generating random field elements from the statement
of the theorem. A deterministic version of this result is also provided
in~\cite{lecerf10}, with a slightly higher arithmetic cost:
$O((d_xd_y)^{(\theta+1)/2})$ instead of $O((d_xd_y)^{1.5})$.
The univariate factorizations in Theorem~\ref{lecerf} can be viewed as
an analogue of the eigenvalue computations in 
Proposition~\ref{cost_lie}.

For an input polynomial $f$ with coefficients in $K$, it may be the case
that a factorization into products of linear forms exists only in an
extension $\mathbb{K}$ of $K$. We will therefore need to apply
Theorem~\ref{lecerf} to such a field extension, and the arithmetic operations
in Theorem~\ref{lecerf} will also take place in this field extension.
We have already made a similar remark for the algorithm
of Section~\ref{factorization} in Section~\ref{other_lie}.

If the input $f$ has degree $d$, we can take $d_x=d_y=d$ and we conclude
that the algorithm of Section~\ref{bivariate} will make $O(nd^3)$
arithmetic operations. For $d=n$ this is smaller than the arithmetic cost
$O(n^{2\theta})$ in Proposition~\ref{cost_lie}, but the latter bound
becomes smaller if $d$ significantly larger than $n$.

A complete analysis should also take the cost of univariate factorizations
into account. Assume for instance that $K=\mathbb{K}=\qq$.
A polynomial time algorithm for this task was first by given
Lenstra, Lenstra and Lovasz~\cite{lenstra82}.
This remains a relatively expensive task despite several improvements
(see~\cite{aecf} for an exposition and more references).
However, we only need to find the linear factors of $F$ (together with
their multiplicities). This boils down to finding the rational roots
of a univariate polynomial, a task which
(as already pointed out in Section~\ref{diagback}) has an essential quadratic
binary cost (\cite{aecf}, Proposition~21.22).

As an alternative to Theorem~\ref{lecerf} one may use the absolute
factorization algorithm by Ch\`eze and Lecerf~\cite{ChezeLecerf07}.
This algorithm only performs arithmetic operations (no univariate factorization
is involved). Moreover, the number of arithmetic operations is barely higher:
$\tilde{O}(d^3)$ instead of $O(d^3)$, where the $\tilde{O}$ notation
hides logarithmic terms. We refer to~\cite{ChezeLecerf07} for a more precise
statement of their result and a description of the output representation.
One advantage of their algorithm is that all arithmetic operations take place
in the coefficient field $K$ of the input polynomial, even if the factors
only exist in a field extension $\mathbb{K}$.
Finally, we note that their algorithm only applies to squarefree
bivariate polynomials. For a reduction from the general case to the squarefree
case we refer to Section~4.2 of~\cite{lecerf10}.

So far, we have not  addressed the arithmetic cost of converting from black box
representation to dense bivariate representation. As pointed out in
Section~\ref{bb_bivariate}, this can be done by interpolating
each of the $n-2$ bivariate polynomials on a set of size $(d+1)^2$.
There are several ways of doing this at negligible cost compared
to dense bivariate factorization.

First, recall that a univariate polynomial of degree $d$ can be interpolated
from its values at roots of unity in $O(d \log d)$ arithmetic operations
using the Fast Fourier Transform. The cost of univariate interpolation
at an arbitrary set of $d+1$ points is a little higher but remains
$\tilde{O}(d)$, see Section~10 of~\cite{gathen13} for details.

Returning to bivariate polynomials, one option is to use the two-dimensional
FFT. Its cost remains $O(N \log N)$, where $N$ is the number of interpolation
points. Here, $N=(d+1)^2$ and we interpolate on $S \times S$
where $S$ is the set of $(1+d)$-th roots of unity.
Another option is to perform the Kronecker substitution $y=x^{d+1}$
and interpolate the polynomial $g(x)=f(x,x^{d+1})$
using one of the aforementioned univariate methods. 
The coefficients of $f$ can be recovered uniquely from those of $g$.

\subsection{The hyperplane algorithm}

Recall that the algorithm proposed in Section~\ref{hyperplane}
factors an input polynomial
$f(x)=\lambda.l_1(x)^{\alpha_1} \cdots l_k(x)^{\alpha_k}$
from $n-1$ univariate polynomial factorizations.
Each univariate polynomial is of the form $g(t)=f(a+tv_i)$ and its
coefficients must be determined by interpolation at Step~1.
As already mentioned in
Section~\ref{other_bi}, this can be done in $O(d \log d)$ arithmetic
operations by FFT from $d+1$ values of $g$. In order to obtain one value
of $g$ we must compute the coordinates of $a+tv_i$ before calling the black box for $f$, at a cost of $n$ arithmetic operations. Interpolating $g$ therefore
takes $O(d\log d + dn)$ arithmetic operations.
Since we have $n-1$ such polynomials to interpolate,
the arithmtic cost of  interpolations is $O(dn (n+\log d))$. 

The roots $a_1,\ldots,a_{k(n-1)}$ of the $n-1$ univariate polynomials
are used at steps 3 and 4 to determine the zero sets of
the $l_i$. At step 5, the multiplicities of the roots of the first polynomial
yield the exponents $\alpha_i$.

At step 3 we test whether the line $(a_1b)$ is included in $Z(f)$, where
$b$ is one the $k(n-1)$ roots. This is done by evaluating the black box
at a random linear combination $ta_1+(1-t)b$. The coordinates of this point
can be computed in $O(n)$ arithmetic operations.
Repeating this for all the $k(n-1)$ roots takes $O(kn^2)$ operations.

At step 4 we determine
$H_1=\mathrm{Span}(a_1,b_2,\ldots,b_{n-1})$, where the $b_i$ have been
found at Step 3.
Finding an equation for $H_1$ amounts to solving a linear system,
and can be done in $O(n^{\theta})$ arithmetic operations as recalled
in Appendix~\ref{other_lie}.
The combined cost of the determination of $H_1$ at steps 3 and 4 is
therefore $O(kn^2+n^{\theta})$. This is repeated for all the hyperplanes,
at a total cost of $O(k(kn^2+n^{\theta}))$ operations.
We recall that these arithmetic operations may take place in a field
extension.

Finally, at Step~6 we evaluate the product $\prod_{i=1}^k l_i(x)^{\alpha_i}$
at some point $x$ and divide $f(x)$ by this product to determine
the normalization factor $\lambda$.
If we use repeated squaring to evaluate the powers $l_i(x)^{\alpha_i}$,
we can complete Step~6 in $O(k(n+\log d))$ arithmetic operations.
Since $k \leq d$ this is negligible compared to the cost of the univariate
interpolations at Step~1.
The above analysis can be therefore be summarized as follows.
\begin{proposition} \label{hyperplane_cost}
  The algorithm of Section~\ref{hyperplane} obtains a factorization
  of the form $f(x)=\lambda.l_1(x)^{\alpha_1} \cdots l_k(x)^{\alpha_k}$
  using $O(k(kn^2+n^{\theta})+dn (n+\log d))$ arithmetic operations.
  The algorithm also needs to
  compute the roots of $n-1$ univariate polynomials of degree $d$,
  and for one of these polynomials it needs to determine
  the multiplicities of roots.
\end{proposition}

For comparison with the Lie-theoretic algorithm, setting $k=n$
in Proposition~\ref{hyperplane_cost} 
yields a count of $O(n^4+dn (n+\log d))$ arithmetic operations.
If $d$ remains polynomially bounded in $n$, this is always smaller
than the corresponding $O(n^{2\theta}+n^3d)$ bound for the black
box version of the Lie-theoretic algorithm.\footnote{We recall that the
  term $n^{2\theta}$ comes from the computation of the Lie algebra of $f$,
  and the term $n^3d$ from the arithmetic cost of polynomial interpolation.
  The arithmetic computations for these two tasks take place in the coefficient
field $K$ of $f$ rather than in a field extension.}
In the white box model, the arithmetic cost of that algorithm
drops to $O(n^{2\theta})$. The Lie-theoretic algorithm therefore
becomes preferable from the point
of view of the arithmetic cost when $d$ exceeds $n^{2\theta-2}$.

\small

\small

{  \section*{Acknowledgements} P.K. would like to thank Gilles Villard
  for useful pointers to the literature on computational linear algebra.}

\bibliographystyle{plain}

\end{document}